\documentclass{article}
\newif\ifproofs\proofsfalse

\usepackage{natbib}
\usepackage{fullpage}
\usepackage{hyperref}
\usepackage{url}
\usepackage{graphicx}
\usepackage{amsmath,amssymb}
\usepackage{amsthm}
\usepackage{bbm}
\usepackage{fullpage}
\usepackage{bm}
\usepackage{mathrsfs}
\usepackage{graphicx}
\usepackage{pifont}
\usepackage{dsfont}
\usepackage{xcolor}
\usepackage{authblk}
\usepackage{soul}

\newtheorem{thm}{Theorem}

\newtheorem{definition}{Definition}
\newtheorem{cor}{Corollary}

\newtheorem{clm}{Claim}
\newtheorem{assumption}{Assumption}

\usepackage{algorithm}
\usepackage{algorithmicx}
\usepackage{algpseudocode}
\usepackage[stable]{footmisc}

\algtext*{EndWhile}
\algtext*{EndFor}
\algtext*{EndIf}
\algnewcommand\algorithmicinput{\textbf{Input:}}
\algnewcommand\INPUT{\item[\algorithmicinput]}
\algnewcommand\algorithmicoutput{\textbf{Output:}}
\algnewcommand\OUTPUT{\item[\algorithmicoutput]}

\DeclareMathOperator*{\Var}{Var}

\renewcommand{\epsilon}{\varepsilon}

\newcommand{\w}{\bold{w}}

\newcommand{\outopt}{o}
\newcommand{\tw}{\tilde{w}}
\newcommand{\htau}{\hat{\tau}}

\newcommand{\he}[1]{\textcolor{blue}{[Hadi: #1]}}
\newcommand{\vg}[1]{\textcolor{blue}{[Vasilis: #1]}}
\newcommand{\jz}[1]{\textcolor{red}{[Juba: #1]}}
\newcommand{\ep}[1]{\textcolor{green}{[Manolis: #1]}}

\renewcommand{\he}[1]{}
\renewcommand{\vg}[1]{}
\renewcommand{\jz}[1]{}
\renewcommand{\ep}[1]{}

\title{Optimal Data Acquisition with Privacy-Aware Agents}

\author[1]{Rachel Cummings}
\author[2]{Hadi Elzayn}
\author[3]{Vasilis Gkatzelis}
\author[3]{\\Emmanouil Pountourakis}
\author[4]{Juba Ziani}

\affil[1]{Columbia University}
\affil[2]{Meta}
\affil[3]{Drexel University}
\affil[4]{Georgia Institute of Technology}

\begin{document}

\maketitle

\begin{abstract}
We study the problem faced by a data analyst or platform that wishes to collect private data from privacy-aware agents. To incentivize participation, in exchange for this data, the platform provides a service to the agents in the form of a statistic computed using all agents' submitted data. 
The agents decide whether to join the platform (and truthfully reveal their data) or not participate by considering both the privacy costs of joining and the benefit they get from obtaining the statistic. The platform must ensure the statistic is computed differentially privately and chooses a central level of noise to add to the computation, but can also induce personalized privacy levels (or costs) by giving different weights to different agents in the computation as a function of their heterogeneous privacy preferences (which are known to the platform). We assume the platform aims to optimize the accuracy of the statistic, and must pick the privacy level of each agent to trade-off between i) incentivizing more participation and ii) adding less noise to the estimate. 

We provide a semi-closed form characterization of the optimal choice of agent weights for the platform in two variants of our model. In both of these models, we identify a common nontrivial structure in the platform's optimal solution: an instance-specific number of agents with the least stringent privacy requirements are pooled together and given the same weight, while the weights of the remaining agents decrease as a function of the strength of their privacy requirement. We also provide algorithmic results on how to find the optimal value of the noise parameter used by the platform and of the weights given to the agents. 
\end{abstract}

\section{Introduction}

Recent advancements in machine learning algorithms and large-scale computation has reaffirmed the crucial value of information, leading to unprecedented levels of data gathering. For example, the recommendation systems used by platforms like Netflix, TikTok, or YouTube are trained on massive amounts of data regarding user behavior and preferences. However, this accumulation of information has raised important concerns regarding the privacy costs suffered by the users that this information pertains to. To mitigate this issue, a large body of research has focused on designing algorithms that process the sensitive user information while limiting their incurred privacy costs (a prominent line of such work focuses on \emph{differential privacy}). The main limitation of this approach is that the reduced privacy costs often come at the expense of lower quality outcomes (e.g., a recommendation system with very strong privacy guarantees may yield poor recommendations), hurting the same users that it is aiming to protect. Our goal in this paper is to develop a better understanding of the trade-offs that such users face between privacy costs and the resulting quality of service, and to design optimal data acquisition mechanisms which respect the users' preferences.


The privacy cost that agents may suffer by releasing access to their data is well-studied: e.g., users in online platforms and social media applications may not want to reveal their search and watch histories or content preferences. In response to these concerns, many platforms allow their users to opt out of sharing their data (e.g., a YouTube user can opt out of letting the platform track their activity). Another common example where collecting sensitive information may be really valuable are medical studies aimed at developing a better understanding of some rare disease. Clearly, an individual that has this disease may be reluctant to share this information, and prefer not to participate in the study.

What may not be as well-understood in these examples is the impact of the potential non-monetary benefits that the agents can accrue by contributing their data: 
e.g., a prime motivation to participate in a medical study about a rare disease is the hope that it may lead to new treatments, which would directly benefit those who suffer from it. Similarly, by revealing their content preferences, users of platforms like Netflix or Youtube can help these platforms improve their recommendation engines which, in turn, provides the users with a higher quality of service. In general, the more significant the potential benefits are, the less reluctant the agents are to share their data. Furthermore, it is often the case that these benefits increase (and privacy costs drop) when more users participate, giving rise to interesting and, to the best of our knowledge, less well-understood complementarity phenomena across users.

Our goal in this paper is to model and analyze such settings in which agents can decide whether or not to release access to their data, by considering both the benefit they would obtain and the privacy losses they would incur. Specifically, we approach this problem from the perspective of a platform whose goal is to maximize the value of the final computation (e.g., the quality of service of a system or the accuracy of a study), while respecting the preferences of the agents. To achieve this goal, the platform can determine the extent to which it will introduce differential privacy protections, taking into consideration the agents' preferences and aiming to incentivize their participation.

\subsection{Summary of contributions} 

In this paper, for simplicity, we consider a learner (e.g., a platform) that is interested in performing a simple task: estimating the mean of a population distribution. The learner controls two types of variables: i) the amount of noise $\eta$ centrally added to the computation for differential privacy, and ii) the weights $w_1, \ldots, w_n$ given to the data of agents $1, \ldots, n$ in the learner's computation. By giving different weights to different agents, the platform can provide personalized privacy levels to agents with varying privacy attitudes; see the preliminary Section~\ref{sec:prelims} for more details on how the privacy level obtained by an agent $i$ depends on $\eta$ and $w_i$. We are interested in designing the weights $\vec{w}$ and the noise $\eta$ to optimize the accuracy of the learner's statistic. 

Ultimately, the accuracy of this statistic will depend on the participation decisions of the agents: as more agents participate, the learner is able to collect more data and to refine his statistic. The first main contribution of our paper is to propose two potential models of how agents decide whether to participate in the platform:
\begin{itemize}
    \item In Section~\ref{sec:og_model_description}, we introduce the ``quasi-linear'' agent model. In this model, an agent explicitly trades-off the privacy losses they incur with the benefit they get from the platform. They only decide to participate in the platform if the anticipated benefit is higher than the cost for sharing their data. A version of this model is also used in the follow-up work of~\citet{ec2022}. 
    \item In Section~\ref{sec: variant_model_description}, we introduce a simpler variant of our model, called the ``privacy-constrained'' model. In this model, an agent is willing to join the platform as long as i) they get some benefit from it and ii) a minimum privacy requirement (that may be different for different agents) is met. 
\end{itemize}

We then proceed to characterizing the optimal choice of estimator (i.e. of weights $\vec{w}$ and noise parameter $\eta$) in semi closed-form:
\begin{itemize}
\item In Section~\ref{sec:og_model}, we do so for the quasi-linear agent participation model. There, we remark that the optimal solution has a non-trivial structure, similar to that of~\citet{acq2}: namely, the agents with the least stringent privacy requirements are pooled together and given the same weight, while agents with higher privacy requirements are given weights that decrease with the strength of their privacy attitudes. We also provide algorithmic guidance on how to find the optimal value of $\eta$. We note that how to elicit the agents' privacy costs when they are strategic and can misreport these costs is studied in~\citet{ec2022}. 
\item In Section~\ref{sec:variant_model}, we show that a similar structure (a pooling region followed by a decreasing weight as privacy attitudes become more stringent) arises in the alternative, privacy-constrained model. We provide expressions for both $w_1, \ldots, w_n$ and $\eta$ nearly in closed-form, up to a single unknown parameter $t$ which controls the number of agents that are pooled together. We also remark that this variant of the model has simple incentive properties: it is in the agents' best interest to report their privacy costs truthfully, even without interventions or payments by the learner.
\end{itemize}

\section{Related work}

Recently, there has been a lot of interest in the study of data transactions in the computer science, operations research, and economics literatures. For example,~\citet{market1,market2} study how to model and design data markets. 

Much of the literature aims to tackle one major building block for data transactions: deciding how to efficiently and optimally acquire data from a collection of agents (or ``data providers''). The main focus of this literature is on settings in which the data providers must be compensated for their data. For example, \citet{acq1,acq2,acq3,acq4,acq5,acq6,acq7,acq8} look at the pricing and purchase of such data when the provided data is verifiable (but providers may be strategic and lie about their costs for revealing their data). There is also a significant line of work--such as~\citet{nonver1,nonver2,nonver3,nonver4,nonver5,nonver6,nonver7,nonver8,nonver9}--on the case of non-verifiable data points, where providers can also lie about their data in order to steer the learner's model towards desired outcomes. 

A significant part of this literature singles out privacy loss as the main reason why data providers must be compensated for their data. This gave rise to a body of work that focuses on data acquisition under differential privacy constraints, e.g. \citet{dp1,dp2,dp3,dp4,dp5,dp6,dp7,dp8}. I.e., the seller must provide formal privacy guarantees on how the providers' data is used, while often still compensating them for any remaining privacy losses. This is where our work lies; we adopt the same point of view as~\citet{dp2,dp8} in that we consider settings in which agents have an inherent interest in the statistic or service offered by the platform that is trained on their data, rather than solely in the payments they receive from the platform.

One of the main, salient elements of our model is that the quality of the estimate or service provided by the platform depends not only on the privacy level that the platform offers, but also on the number of providers that join the platform and report their data. In turn, the agents' participation decisions are an endogenous aspect of our model, as in the works of \citet{arpita-katrina} and \citet{ec2022}. Similarly to our setting, \citet{arpita-katrina} consider a setting in which the privacy cost a data provider incurs depends on other providers' participation decision; the main distinction compared to our work is that in \citet{arpita-katrina}, agents only care about how much privacy they obtain, not on how the collected data is used by the buyer or platform to offer a useful service or statistic in return. The work of~\citet{ec2022}, follows-up on an earlier working version of our model~\cite{INFORMS-talk}, and extends it 
to also consider settings where agents may strategically misreport their privacy costs. Specifically, they consider a data acquisition setting with verifiable data where agents obtain a benefit that depends on the accuracy of the platform's model, rather than only from payments they get from their data. 
We remark, however, that our results are mostly orthogonal to theirs. They focus on designing algorithms for solving the mechanism design problem, i.e., how to incentivize truthful cost reporting while optimizing the accuracy of the platform's estimate. On the other hand, we focus on characterizing how the optimal estimator for the platform should weigh each agent's data in semi- closed form, as a function of their privacy preferences. We also incorporate several additional modeling elements relative to~\citet{ec2022}: in particular, i) we assume that each agent may benefit from the platform's estimation in a possibly non-linear way, and ii) we provide a second, alternative model of agents' privacy preferences and of how they decide to participate in the platform.

\section{Differential Privacy Preliminaries}
In this paper, we focus on \emph{differential privacy} as our main privacy technique. Differential privacy was first introduced in the seminal work of~\citet{dp_og} and aims to prevent an attacker from being able to infer an agent's data by observing or post-processing the output of an algorithm, e.g. the output of a learner's statistical computation or machine learning model. In this section, we focus on presenting the minimal knowledge of differential privacy needed for this paper; for a more detailed discussion of differential privacy, please refer to~\citet{dworkroth}.

Differential privacy protects an agent's data by comparing two possible worlds for each agent; the difference between these two worlds is that they consider two possible different values for the data of this agent. Differential privacy requires that one (almost) cannot distinguish between these two worlds by looking at the (distribution over) outputs of the learner's computation; i.e., one  cannot tell with any reasonable certainty what the data point of the agent was, since the outcome of the computation (nearly) does not depend on its value. Formally, a learner runs a computation of mechanism $\mathcal{M}$ which takes a dataset $x$ as an input, and outputs some function or property $\mathcal{M}(x)$ of that dataset. Given $n$ agents whose data is used in the learner's mechanism, one can think of a dataset $x$ as a vector of entries $(x_1,\ldots,x_n)$, where $x_i$ is the data of agent $i$. We first introduce the definition of neighboring datasets:

\begin{definition}
Two datasets $x$ and $x'$ are neighboring with respect to agent $i$ (or ``$i$-neighbors'') if they differ only in agent $i$'s data. I.e., $x_j = x_j'$ for all $j \neq i$.
\end{definition}

Differential privacy, as informally described above, requires that the outputs of mechanism $\mathcal{M}$ differ little on any two neighboring databases $x$ and $x'$. This is formalized as follows:
\begin{definition}[$\varepsilon$-differential privacy]
Let $\varepsilon > 0$. A randomized algorithm $\mathcal{M}$ is $\varepsilon$-differentially private with respect to agent $i$ if for any outcome set $O \subset Range\left(\mathcal{M}\right)$ and for all neighboring databases $x, x'$ with respect to $i$,
\[
\Pr \left[\mathcal{M}(x) \in O \right] \leq \exp(\varepsilon) \Pr \left[\mathcal{M}(x') \in O \right].
\]
\end{definition}
Here, the parameter $\varepsilon$ controls how much privacy each agent gets. As $\varepsilon$ decreases, $\exp(\varepsilon)$ also decreases and the above constraint becomes more and more stringent, improving the level of privacy guaranteed by the mechanism. For $\varepsilon = 0$, it in fact requires that $\Pr \left[\mathcal{M}(x) = o \right] = \Pr \left[\mathcal{M}(x') = o \right]$; i.e., the outcome of the mechanism is independent of the input data and thus perfectly preserves privacy. As $\varepsilon \to +\infty$, the above constraint is trivially satisfied by any mechanism and no privacy protection is provided. 

One of the simplest way to answer a desired statistical query in a differentially private manner is to add noise to the output of said query. Intuitively, as the amount of added noise \emph{increases}, the dependency of the result of said query on any particular agent's data \emph{decreases} (equivalently, the level of privacy obtained by agents \emph{increases}). The most basic and common mechanism to obtain differential privacy answers to numerical queries is the Laplace mechanism, which adds Laplace noise to the output of a query. 

\begin{definition}
Let $q$ be a numerical query, i.e. $q(x) \in \mathbb{R}$ for all $x$. The Laplace mechanism is defined as
\[
\mathcal{M}_L(x,q,\eta) = q(x) + Z,
\]
where $Z$ is a random variable drawn from the Laplace distribution with parameter $\eta$.
\end{definition}

The level of privacy obtained by the Laplace mechanism depends on the sensitivity of the query we aim to answer; i.e., how much the value of this query changes when the data entry of a single agent in the database changes. Formally, the sensitivity of a query with respect to agent $i$ is defined as 
\[
\left(\Delta q\right)_i = \max_{x,x'~\text{i-neighbors}} |q(x) - q(x')|.
\]
We then have the following privacy guarantee for agent $i$:
\begin{definition}
$\mathcal{M}_L(x,q,\varepsilon) = q(x) + Z$ is $\eta \left(\Delta q\right)_i$-differentially private with respect to agent $i$.
\end{definition}

Finally, we note that our goal is to both provide individual agents reporting their data with privacy guarantees while at the same time obtaining an accurate estimate of the statistic we are interested in. Because we consider unbiased estimators in this paper, the accuracy of said estimator is directly linked to its variance. The variance of the Laplace mechanism with parameter $\eta$ on query $q$ is given by
\[
\Var \left(q(x) + Z\right) = \Var{}_x \left(q(x)\right) + \frac{2}{\eta^2}. 
\]

\label{sec:prelims}

\section{Model}

We model a setting in which a data analyst or platform aims to incentivize privacy-aware agents to share their data with or join the platform, then collects their data and uses it to estimate a statistic. To incentivize agent participation, the platform simultaneously aims to provide privacy guarantees to agents who join the platform  while also offering a useful service to the agents who join  through their machine learning model. E.g., the machine learning model could be a platform's recommendation system, such as the ones offered by platforms such as YouTube and TikTok; it could also be the product of a medical study on a rare disease, where the individuals contribute their sensitive medical data to a study in the hopes of getting better treatments and medical outcomes in return.

The platform faces a population of $n$ agents. Each agent has a private data point $d_i$. The data points are drawn i.i.d. from an unknown distribution with unknown mean $\mu$ but known variance $\sigma^2$. Each agent also has a linear privacy cost function given by $c_i \varepsilon_i$, where $c_i \geq 0$ is an agent-specific scalar and $\epsilon_i > 0$ is the level of differential privacy obtain by agent $i$ if he joins the platform; this linearity assumption follows that of \citet{dp5}. 

The goal of the platform is to i) incentivize agents to join the platform, then ii) compute an unbiased estimator $\hat{\mu}$ of $\mu$\footnote{We see this simple estimation task as a possible proxy for more complex machine learning tasks (such as training a recommendation system), and leave the study of such tasks to future work.}. The platform wants this estimator to be as accurate as possible. Letting $S \in [n]$ be the set of agents that decide to join the platform, we assume that the platform's estimator is linear, i.e. given by
\[
\hat{\mu}(S,\w,\eta) = \sum_{i \in S} w_i d_i + Z(\eta),
\]
where $w_i$ is the weight assigned to the data of agent $i$ and $Z$ is a random variable drawn from a Laplace distribution with parameter $\eta \geq 0$ for privacy. We denote as $\w$ the vector of all $w_i$'s. Since we require our estimator to be unbiased, we assume that $w_i \geq 0$ for all $i \in [n]$ and that $\sum_{i \in S} w_i =1$. The platform optimizes over both the choice of weights $\{w_i\}_{i \in S}$ and of noise parameter $\eta$. 

Because the estimator used by the platform is unbiased, we can measure its performance (here, its expected mean-squared error) through its variance. The variance of $\hat{\mu}$, as per preliminary section~\ref{sec:prelims}, is given by
\[
\Var(\hat{\mu}) = \sum_{i \in S} w_i^2 \sigma^2 + \frac{2}{\eta^2}.
\]

The order of operations is then the following:
\begin{enumerate}
\item The analyst announces the weight vector $\w$ and the noise parameter $\eta$ that she will use in the computation.
\item Each agent $i$ decides whether he wants to participate given $\w, \eta$.
\item The analyst computes the estimator $\hat{\mu}$ on the participating agents.
\end{enumerate}

We propose two variants on how we model agents' privacy attitudes, utilities, and participation decisions. In the ``quasi-linear agent'' model, agents maximize a quasi-linear utility functions that trades-off the quality of the final model and their privacy costs. In the ``privacy-constrained agent'' model, agents aim to maximize the utility they get from the platform's model under a constraint that their privacy is not violated by more than a desired tolerance.

\subsection{The Quasi-Linear Agent Model}\label{sec:og_model_description}

In the quasi-linear model, agent $i$ has a quasi linear utility for participating in the platform, which trades-off his privacy cost for reporting his data and his utility from the platform's estimation. Noting that the sensitivity of  estimator $\hat{\mu}$ with respect to agent $i$ is given by $w_i$, the level of privacy obtained by agent $i$ is given by $\epsilon_i = w_i \eta$ (as discussed in Section~\ref{sec:prelims}), and $i$ incurs cost $c_i w_i \eta$ for participating in the platform. In turn, we consider the following quasi-linear utility for the agent: 
\[
u_i(\w,\eta) = f \left(\sigma^2 \sum_{i \in S} w_i^2 + \frac{2}{\eta^2}\right) - c_i w_i \eta
\]
for some decreasing function $f$; i.e., agent $i$'s utility increases when the variance of the platform's model decreases, and when his privacy cost decreases. If the agent does not join the platform, we assume they have access to an outside option with utility $\outopt$ (for example, they could use their own data point as an estimate). The agent then decides to participate if and only if
\[
f \left(\sigma^2 \sum_{i \in S} w_i^2 + \frac{2}{\eta^2}\right) - c_i w_i \eta \geq \outopt.
\]
The platform then aims to solve the following optimization problem:
\begin{align}
\begin{split}
\min_{\eta,\w,S}~& \sum_{i \in S} w_i^2 \sigma^2 + \frac{2}{\eta^2}
\\\text{s.t.}~& c_i w_i \eta \leq f \left(\sigma^2 \sum_{i \in S} w_i^2 + \frac{2}{\eta^2}\right) -\outopt ~\forall i \in S
\\& \sum_{i \in S} w_i = 1
\\&w_i \geq 0~\forall i
\end{split}
\end{align}
where the first constraint ensures that agents in $S$ choose to participate, and the last two constraints enforce that the platform's estimator is unbiased. 

Finally, we make the following assumption that function $f$ is well-behaved for our purposes:
\begin{assumption}
$f$ is concave and differentiable.
\end{assumption}

\subsection{The Privacy-Constrained Agent Model}\label{sec: variant_model_description}

We now consider a variant of our model of agent behavior. In the ``privacy-constrained agent model'', each agent $i$, on top of a privacy cost, also has a \emph{privacy budget} $B_i$ which is the maximum privacy cost the agent is willing to incur. An agent's utility for participation is then given by 
\begin{align*}
u_i(\w,\eta) 
= 
\begin{cases}
g \left(\sum_{i \in S} w_i^2 \sigma^2 + \frac{2}{\eta^2}\right)~~&\text{if } c_i w_i \eta \leq B_i
\\ - \infty &\text{otherwise},
\end{cases}
\end{align*}
where $g$ is a non-negative (agents get more utility from participating than non-participating) and decreasing function. I.e., an agent is never willing to participate if his privacy budget is violated. Otherwise, if the agent's privacy requirement is met, his utility is given by a function of the accuracy of the model. The analyst's program is then given by:
\begin{align}
\begin{split}
\min_{\eta,\w,S}~& \sum_{i \in S} w_i^2 \sigma^2 + \frac{2}{\eta^2}
\\\text{s.t.}~& c_i w_i \eta \leq B_i~~\forall i \in S
\\& \sum_{i \in S} w_i = 1
\\&w_i \geq 0~\forall i
\end{split}
\end{align}
Note that in this case, each agent gets utility $g \left(\sum_{i \in S} w_i^2 \sigma^2 + \frac{2}{\eta^2}\right)$, and minimizing the variance of the platform's estimate also maximizes the agents' utilities. We can further re-write the program as
\begin{align}
\begin{split}
\min_{\eta,\w,S}~& \sum_{i \in S} w_i^2 \sigma^2 + \frac{2}{\eta^2}
\\\text{s.t.}~& w_i \eta \leq \tau_i~~\forall i \in S
\\& \sum_{i \in S} w_i = 1
\\&w_i \geq 0~\forall i,
\end{split}
\end{align}
where $\tau_i \triangleq \frac{B_i}{c_i}$ is called the \emph{privacy threshold} of agent $i$. We assume $\tau_i > 0$ for all $i$; agents with $\tau_i = 0$ require $w_i = 0$, do not affect the objective function, and can be dropped from the computation without loss of generality.


This model is a more tractable variant of the quasi-linear one in that participation decisions by the agents are significantly simplified. Even when the value of $\eta$ used by the platform is known or announced, the ``quasi-linear'' model considers agents that trade-off their privacy losses with their benefit from the platform's model. In this case, an agent's participation decision depends on them being able to anticipate the quality of the final model, which requires access to the weights given to other agents. To do so, the platform either needs to communicate these weights to the agent, or each agent can solve the optimization himself, which may require unrealistic knowledge about the other agents' costs as well as unrealistic reasoning and computational power. In contrast, an agent in the ``privacy-constrained'' setting makes a simpler decision that only depends on his own weight $w_i$ (this can be interpreted as a promise to the agent on how much their data is going to be used at most) and their privacy preferences $\tau_i$. In Section \ref{sec:variant_model}, we will note that despite its relative simplicity, the ``privacy-contrained'' model offers similar insights to that of the ``quasi-linear'' model in Section~\ref{sec:og_model}; this provides evidence that even this simplified model can provide valuable guidance on how to acquire and use data from agents with heterogeneous privacy preferences.

\section{Characterizing the Optimal Solution under the ``Quasi-Linear'' Agent Model}\label{sec:og_model}

Recall that the optimization problem solved by the platform is given by 
\begin{align}
\begin{split}\label{eq:weighted_optim}
\min_{\eta,\w,S}~& \sum_{i \in S} w_i^2 \sigma^2 + \frac{2}{\eta^2}
\\\text{s.t.}~& c_i w_i \eta \leq f \left(\sigma^2 \sum_{i \in S} w_i^2 + \frac{2}{\eta^2}\right) -\outopt ~\forall i \in S
\\& \sum_{i \in S} w_i = 1
\\&w_i \geq 0~\forall i
\end{split}
\end{align}

\paragraph{Re-writing the optimization problem} We first rewrite the optimization problem solved by the analyst in a simpler form. To do so, we show how to drop the dependency of the optimization program in $S$. We now only need to optimize over $\vec{w}$ and $\eta$.

\begin{clm}\label{clm: optim_final}
Consider the following program:
\begin{align}
\begin{split}\label{eq:weighted_optim_final}
\min_{\eta,\w}~& \sigma^2 \sum_{i = 1}^n w_i^2+ \frac{2}{\eta^2}
\\\text{s.t.}~& c_i w_i \eta \leq f \left(\sigma^2 \sum_{i = 1}^n w_i^2 + \frac{2}{\eta^2}\right) - \outopt~\forall i \in [n]
\\& \sum_{i = 1}^n w_i = 1
\\&w_i \geq 0~\forall i.
\end{split}
\end{align}
Let $S = \{i~\text{s.t.}~w_i > 0\}$. Then $(\w,\eta,S)$ is an optimal solution to Program~\eqref{eq:weighted_optim} if and only if $(\w,\eta)$ is an optimal solution to Program~\eqref{eq:weighted_optim_final}, and both programs have the same optimal value. 
\end{clm}

\begin{proof}[Proof Sketch]
Clearly $(\w,\eta,S)$ yields the same objective value for Program~\eqref{eq:weighted_optim} as $(\w,\eta)$ does for Program~\eqref{eq:weighted_optim_final}. Further, $(\w,\eta,S)$ is feasible for Program~\eqref{eq:weighted_optim} if and only if $(\w,\eta)$ is feasible for Program~\eqref{eq:weighted_optim_final}. Both statements put together imply that both programs have the same optimal value and that said optimal value is reached at $(\w,\eta,S)$ and $(\w,\eta)$ respectively. More details are provided in Appendix~\ref{app:proof_claims_quasilinear}.
\end{proof}

In short, note that if we find an optimal solution to Program~\ref{eq:weighted_optim_final}, we can construct an optimal solution to Program~\eqref{eq:weighted_optim} with the same objective value. Studying Program~\ref{eq:weighted_optim_final} is without loss of generality. 

Note however that the above optimization problem may be hard to solve directly as it is not convex: indeed, $(\w,\eta) \rightarrow w_i \eta$ is not a jointly convex function of $\w$ and $\eta$. To deal with this issue, we note that once we fix the value of $\eta$, the problem is now entirely convex. Indeed, i) the objective function is convex in $\w$, ii) $c_i w_i \eta$ is convex in $\w$ and $-f \left(\sigma^2 \sum_{i = 1}^n w_i^2 + \frac{2}{\eta^2}\right)$ is convex in $\w$ (because -f is convex increasing and $\sigma^2 \sum_{i = 1}^n w_i^2 + \frac{2}{\eta^2}$ is convex), and iii) the weight constraints are linear. In this section, we mostly focus on understanding this convex optimization problem for any fixed value of $\eta$. Finding the best $\eta$ corresponds to finding the optimum of a one-dimensional function, which can be approximated heuristically through black-box optimization techniques.

\paragraph{Properties of the optimal solution} In the rest of this section, we order agents as a function of their privacy costs. I.e., without loss of generality, we number agents such that $c_1 \leq \ldots \leq c_n$. As mentioned above, we now consider optimization at fixed $\eta$. I.e., for any given $\eta$, we aim to solve program
\begin{align}
OPT(\eta) =
\min_{\w}~& \sigma^2 \sum_{i = 1}^n w_i^2+ \frac{2}{\eta^2} \notag
\\\text{s.t.}~& c_i w_i \eta \leq f \left(\sigma^2 \sum_{i = 1}^n w_i^2 + \frac{2}{\eta^2}\right) - \outopt ~\forall i \in [n] \label{eq:weighted_optim_fixedeta}
\\& \sum_{i = 1}^n w_i = 1 \notag
\\&w_i \geq 0~\forall i. \notag
\end{align}

We first note the following simple monotonicity result on the structure of an optimal solution to Program~\eqref{eq:weighted_optim_fixedeta} hence \eqref{eq:weighted_optim_final}:
\begin{clm}\label{clm:monotone_w}
Take any $\eta \geq 0$ such that Program~\eqref{eq:weighted_optim_fixedeta} is feasible, then any optimal solution to Program~\eqref{eq:weighted_optim_fixedeta} satisfies $w_1 \geq \ldots \geq w_n$.
\end{clm}
I.e., as we would intuitively expect, agents with smaller privacy costs get more weight in the computation. This allows the platform to provide more privacy to agents with higher costs to ensure said costs do not become too high and violate the participation constraint. We also note the following monotonicity result of independent interest, which states that the privacy costs of the agents are in fact monotone increasing in the $c_i$'s. I.e., agents with less stringent privacy attitudes also end up incurring lower privacy costs. 
\begin{clm}\label{clm:monotone_equilibrium_costs}
Any optimal solution $\w$ to Program~\eqref{eq:weighted_optim_fixedeta} satisfies $c_i w_i \eta \leq c_j w_j \eta$ for all $i < j$.
\end{clm}

The proofs of the three previous claims are provided in Appendix~\ref{app:proof_claims_quasilinear}. The proofs are by contradiction and show that if an optimal solution satisfies the condition of each of the claim, then we can construct a feasible solution with better objective, contradicting optimality. Putting the previous claims together, we obtain the following corollary:

\begin{cor}\label{cor:full_participation}
Any optimal solution $\w$ satisfies $w_i > 0$ for all $i \in [n]$.
\end{cor}

\begin{proof}
Suppose this were not true, i.e. for some $i$, $w_i = 0$. Then by monotonicity of $\w$ proven in Claim~\ref{clm:monotone_w}, it must be that $w_{n} = 0$. But then, by Claim~\ref{clm:monotone_equilibrium_costs}, it must be that $c_i w_i \leq c_{n} w_{n} = 0$ for all $i$. This implies $w_i = 0$ for all $i$, which contradicts $\sum_i w_i = 1$. 
\end{proof}

We note that in our optimal solution, \emph{every} agent is incentivized to participate in the platform and to report their data. This is the result of a self-reinforcing effect exhibited in our setting: on the one hand, more participation means that the platform computes a more accurate model, which incentivizes more agent participation; on the other hand, more participation lowers the privacy costs of the agents (as it lowers how much the computation depends on any given agent's data), which also helps incentivizing more participation. 

\paragraph{A semi-closed form characterization} We now provide the main characterization result of this section; namely, a semi-closed form solution for Program~\eqref{eq:weighted_optim_fixedeta}. 

\begin{thm}\label{thm:semi_close-form}
Assume Program~\ref{eq:weighted_optim_fixedeta} is feasible. Let $\w$ be any optimal solution to Program~\ref{eq:weighted_optim_fixedeta}. There exists constants $K$ and $W$ and an integer $t$ such that $w_i = W$ for all $i \leq t$ and $w_i = K/c_i$ for all $i \geq t+1$.
\end{thm}

\begin{proof}[Proof sketch]
The full proof of the result relies follows by examining the implications of the Karush–Kuhn–Tucker (KKT) conditions for optimality and is provided in Appendix~\ref{app:proof_mainthm}. One technicality is that the KKT conditions require that Slater's condition holds. This means that the optimization program needs to be strictly feasible, i.e. there must exist a feasible solution such that all inequality constraints are strictly satisfied. To circumvent this issue, we note that when we do not have strict feasibility, any feasible (hence the optimal) solution must make \emph{all} participation constraints tight hence is easy to characterize. 
\end{proof}

We remark that our optimal solution exhibits interesting structure. First, there is a pooling region in which the agents with the lowest privacy costs are given the same weights. Then, agent weights start decreasing in their cost to ensure that their privacy losses do not become too big. We note that this result is in line with that of~\cite{acq2}. This is perhaps surprising given that~\cite{acq2} considers a different objective and constraints for the platform.

A potential explanation may be that absent privacy constraint, the optimal solution in terms of variance is to give the same weight to every agent. However, this may not be possible due to the agents' privacy requirements. Instead, one wants to have a solution that keeps the weights of different agents equal when possible to minimize the variance due to these agents, and only give a different, lower weight to the agents when this is unavoidable to ensure they participate in the computation. 

\paragraph{Finding the optimal value of $\eta$} One possible approach to optimize over the value of $\eta$ is to do a grid search over said 1-dimensional parameter. However, $OPT(\eta)$ is a black-box, not well understood function of $\eta$, that may be complex to optimize over. Another approach is to refine our understanding of the relationship between $K$, $W$, and $\eta$. One way to do so is to first note that if we know $t$, there is a closed-form relationship between $K$ and $W$. In particular, we have that 
\[
t W + K \sum_{i > t} \frac{1}{c_i} = 1,
\]
implying that 
\[
W = \frac{1}{t} \left(1 - K \sum_{i > t} \frac{1}{c_i}\right).
\]
From the proof of Theorem~\ref{thm:semi_close-form} found in Appendix~\ref{app:proof_mainthm}, we also know that the participation constraint is tight for all agents $i > t$ with $w_i = \frac{K}{c_i}$, hence it must be that 
\[
K \eta = f\left(\sigma^2 t W^2 + K^2 \sum_{i > t} \frac{1}{c_i}^2 + \frac{2}{\eta^2}\right) - \outopt.
\]
This can be rewritten as
\begin{align*}
&K \eta = f\left(\frac{\sigma^2}{t} \left(1 - K \sum_{i > t} \frac{1}{c_i}\right)^2 + K^2 \sum_{i > t} \frac{1}{c_j}^2 + \frac{2}{\eta^2}\right) - \outopt.
\end{align*}
In particular, for each possible value of $t$, we can restrict our search to the parameters $K$ and $\eta$ that satisfy the above equation. In the special case where $f$ is linear, this equation is quadratic and has (at most) two well-behaved solutions that depend continuously on the value of $\eta$. This facilitates a grid search approach to find the best $\eta$ for each possible value of $t$. We can then simply pick the value of $t$ that leads to the best objective value.

\section{Characterizing the Optimal Solution under the ``Privacy-Constrained'' Agent Model}\label{sec:variant_model}

Recall that the optimization program solved by the platform is given by:
\begin{align}\label{eq:optim_threshold} 
\begin{split}
\min_{\eta,\w,S}~& \sum_{i \in S} w_i^2 \sigma^2 + \frac{2}{\eta^2}
\\\text{s.t.}~& w_i \eta \leq \tau_i~~\forall i \in S
\\& \sum_{i \in S} w_i = 1
\\&w_i \geq 0~\forall i,
\end{split}
\end{align}

\paragraph{Re-writing the optimization problem} We start by noting that Program~\ref{eq:optim_threshold} can be rewritten in a simpler form involving no $S$ variable. Indeed:

\begin{clm}\label{clm: optim_final_variant}
Consider the following program:
\begin{align}
\begin{split}\label{eq:optim_threshold_final}
\min_{\eta,\w}~& \sum_{i=1}^n w_i^2 \sigma^2 + \frac{2}{\eta^2}
\\\text{s.t.}~& w_i \eta \leq \tau_i~~\forall i \in [n]
\\& \sum_{i = 1}^n w_i = 1
\\&w_i \geq 0~\forall i,
\end{split}
\end{align}
Let $S = \{i~\text{s.t.}~w_i > 0\}$. Then $(\w,\eta,S)$ is an optimal solution to Program~\eqref{eq:optim_threshold} if and only if $(\w,\eta)$ is an optimal solution to Program~\eqref{eq:optim_threshold_final}, and both programs have the same optimal value. 
\end{clm}

\begin{proof}
The proof is nearly identical to that of Claim~\ref{clm: optim_final} and is omitted for the sake of brevity. 
\end{proof}

Once again, this optimization problem is not convex. However, if we fix $\eta$ and only consider $\w$ as a variable, our optimization problem becomes convex. We can then solve the problem efficiently for any desired value of $\eta$, then search over $\eta$ to find the optimal solution. 

In the rest of this section, we first show that the optimal solution has similar positivity and monotonicity properties to that of the ``quasi-linear'' model. We then show that we can characterize the optimal solution in semi-closed form. Finally, we exploit the structure of our problem to provide a simple characterization and algorithm for finding the optimal $\eta$.

\paragraph{Properties of the optimal solution} Without loss of generality, we number agents so that $\tau_1 \geq \tau_2 \geq \ldots \geq \tau_n$. I.e., agents with higher indices have more stringent privacy requirements. As mentioned above, we now consider the optimization at fixed $\eta$, and study the  problem 
\begin{align}
\begin{split}\label{eq:optim_fixed_eta}
\min_{\w}~& \sum_{i =1}^n w_i^2 \sigma^2 + \frac{2}{\eta^2}
\\\text{s.t.}~& w_i \eta \leq \tau_i~~\forall i \in [n]
\\& \sum_{i = 1}^n w_i = 1
\\&w_i \geq 0~\forall i.
\end{split}
\end{align}

To draw a parallel with the ``quasi-linear'' model, we first show that this variant of our model exhibits strong monotonicity and positivity properties. 
\begin{clm}\label{clm:positive_monotone}
Suppose $\w$ is an optimal solution. Then $w_i > 0~~\forall i \in [n]$, $w_1 \geq \ldots \geq w_n$, and $\frac{w_1}{\tau_1} \leq \ldots \leq \frac{w_n}{\tau_n}$.
\end{clm}

\paragraph{A semi-closed form solution} Claim~\ref{clm:positive_monotone} provides a high level understanding of the shape of the optimal solution. We now refine this understanding by providing a semi-closed form solution to Program~\ref{eq:optim_threshold_final}.

\begin{thm}\label{thm:semi_close-form_2}
Let $\w$ be any optimal solution to Program~\ref{eq:optim_fixed_eta} (assuming feasibility). There exists $t \in \{0,\ldots,n\}$ and $W \geq 0$ such that is given by $w_i = W$ for all $i \leq t$ and $w_i = \frac{\tau_i}{\eta}$ for all $i \geq t+1$.
\end{thm}

\begin{proof}[Proof Sketch]
As before, the full proof of the result relies on the Karush–Kuhn–Tucker (KKT) conditions and is provided in Appendix~\ref{app:proof_mainthm_2}. The proof suffers from the same technicality that the KKT conditions require that the optimization program is strictly feasible, and we use the same techniques as for Theorem~\ref{thm:semi_close-form} to circumvent said issue. 
\end{proof}

We note that the above result bears similarities to that of Theorem~\ref{thm:semi_close-form} in the ``quasi-linear'' model. Indeed, note that $\tau_i$ is a parameter that is smaller as the privacy preferences of agent $i$ are more stringent, similarly to $1/c_i$ in the ``quasi-linear'' model. Both solutions then have the same structure: agents with more lax privacy requirements are pooled together and have the same weight, while agents with stronger requirements see their weight decrease as a function of how strong that requirement is. 

\begin{cor}\label{cor:close_form_w}
There exists $t \in \{0,\ldots,n\}$ such that the optimal solution to Program~\ref{eq:optim_fixed_eta} (assuming feasibility) is given by $w_i = \frac{1}{t} \left(1 - \frac{1}{\eta} \sum_{i = t+1}^n \tau_i \right)$ for all $i \leq t$ and $w_i = \frac{\tau_i}{\eta}$ for all $i \geq t+1$.
\end{cor}

\begin{proof}
It suffices to use the fact that the weights of the agents must sum to $1$. I.e., $\sum_{i=1}^n w_i = 1$ can be rewritten as 
\[
\sum_{i=1}^t W + \sum_{i = t+1}^n \frac{\tau_i}{\eta} = 1,
\]
or equivalently 
\[
t W + \frac{1}{\eta} \sum_{i = t+1}^n \tau_i = 1.
\]
This immediately leads to $W = \frac{1}{t} \left(1 - \frac{1}{\eta} \sum_{i = t+1}^n \tau_i \right)$.
\end{proof}

\paragraph{Finding the optimal value of $\eta$ exactly} We show that, in fact, $\eta$ can be found by simply minimizing a function of a single variable on a closed interval:
\begin{clm}
The optimal value of $\eta$ is given by
\begin{align}
\eta^* = \arg\min_\eta h(\eta)~\text{s.t.}~ \eta \in \left[\sum_{i=1}^{t+1} \tau_i,t \tau_t + \sum_{i=1}^{t+1} \tau_i\right],
\end{align}
where
\[
h(\eta) 
= \frac{\sigma^2}{t^2} \sum_{i = 1}^t \left(1 - \frac{1}{\eta} \sum_{i = t+1}^n \tau_i \right)^2 + \frac{\sigma^2}{\eta^2} \sum_{i= t + 1}^n \tau_i^2 + \frac{2}{\eta^2}.
\]
\end{clm}

\begin{proof}
Assuming the optimal value of $t$ is known, plugging the solution of Corollary~\ref{cor:close_form_w} back into Program~\eqref{eq:optim_threshold} shows that $\eta$ must solve
\begin{align*}
\min_{\eta \geq 0}~& \frac{\sigma^2}{t^2} \sum_{i = 1}^t \left(1 - \frac{1}{\eta} \sum_{i = t+1}^n \tau_i \right)^2 + \frac{\sigma^2}{\eta^2} \sum_{i= t + 1}^n \tau_i^2  + \frac{2}{\eta^2}
\\\text{s.t.}~& \frac{1}{t} \left(1 - \frac{1}{\eta} \sum_{i = t+1}^n \tau_i \right) \eta \leq \tau_i~~\forall i \leq t
\\&\frac{1}{t} \left(1 - \frac{1}{\eta} \sum_{i = t+1}^n \tau_i \right) \geq 0 
\end{align*}
Note that we dropped the constraint that the weights sum to $1$: this is guaranteed to hold for any plugged-in solution of the form given in Corollary~\ref{cor:close_form_w}. Using the fact that $\tau_1 \geq \ldots \geq \tau_t$, we can rewrite the problem as
\begin{align}
\begin{split}
\min_{\eta \geq 0}~& \frac{\sigma^2}{t} \left(1 - \frac{1}{\eta} \sum_{i = t+1}^n \tau_i \right)^2 + \frac{\sigma^2}{\eta^2} \sum_{i= t + 1}^n \tau_i^2 + \frac{2}{\eta^2}
\\\text{s.t.}~& \frac{1}{t} \left(1 - \frac{1}{\eta} \sum_{i = t+1}^n \tau_i \right) \eta \leq \tau_t,
\\&\frac{1}{t} \left(1 - \frac{1}{\eta} \sum_{i = t+1}^n \tau_i \right) \geq 0,
\end{split}
\end{align}
or equivalently 
\begin{align}
\begin{split}\label{eq:optim_eta}
\min_{\eta}~& \frac{\sigma^2}{t} \left(1 - \frac{1}{\eta} \sum_{i = t+1}^n \tau_i \right)^2 + \frac{\sigma^2}{\eta^2} \sum_{i= t + 1}^n \tau_i^2 + \frac{2}{\eta^2}
\\\text{s.t.}~& \eta \leq t \tau_t + \sum_{i=t + 1}^n \tau_i,
\\&\eta \geq \sum_{i=1}^{t+1} \tau_i.
\end{split}
\end{align}
\end{proof}

We further show that this is a simple optimization problem in that $f(\eta)$ is well-behaved and easy to minimize.

\begin{clm}\label{clm:simple_optim_eta}
There exists $\eta^*$ such that $f(\eta)$ is decreasing for $\eta < \eta^*$, and increasing for $\eta > \eta^*$. In turn, $\eta^*$ is the unique solution to $f'(\eta) = 0$, and is given in closed form by 
\[
\eta^* = \frac{\left(\sum_{i = t+1}^n \tau_i\right)^2 + t \sum_{i = t+1}^n \tau_i^2 + \frac{2t}{\sigma^2}}{\sum_{i = t+1}^n \tau_i}.
\]
If $\eta^* \in [\sum_{i=1}^{t+1} \tau_i,t \tau_t + \sum_{i=1}^{t+1} \tau_i]$, it minimizes $f$; otherwise, the minimizer is either $\sum_{i=1}^{t+1} \tau_i$ or $t \tau_t + \sum_{i=1}^{t+1} \tau_i$.
\end{clm}

The proof follows from simple algebra and is given in Appendix~\ref{app:simple_optim_eta}. We note that the above characterization gives us an immediate algorithm to find the optimal $\eta$. Indeed, it suffices to explore all $n$ possible values of $t$. For each $\eta$, then, one only has to compare $f(\eta^*)$, $f\left(\sum_{i=1}^{t+1} \tau_i\right)$ and $f \left(t \tau_t + \sum_{i=1}^{t+1} \tau_i\right)$. This algorithm takes time $O(n)$.

\paragraph{Incentive properties of the ``privacy-constrained'' model} Finally, we remark that the privacy-constrained model enjoys nice incentive properties: e.g., it is a weakly dominated strategy for agents to misreport their privacy thresholds. We note that this property holds \emph{without having to pay agents to report their privacy preferences truthfully}. This is a major advantage in that it reflects what happens in real-life platforms, who often do not pay their users to incentivize participation; in fact, platforms commonly ask users to pay to be able to access the service they offer in return. We divide the incentive properties in the following two claims:

\begin{clm}
For any agent $i$, reporting $\htau_i < \tau_i$ is a weakly dominated strategy. 
\end{clm}

\begin{proof}
Fix the participation strategy of all the other agents--let $S_{-i}$ the set of agents that decide to join the platform and report their data--, as well as their reports $\htau_j$ for all $j \in S_{-i}$. Suppose the set of participating agents is $S = S_{-i} \cup {i}$ (i.e. $i$ decides to participate) and $\htau_i < \tau_i$. Let $OPT\left(S,\bold{\htau}\right)$ be the optimal objective value of Program~\eqref{eq:optim_threshold_final} when the inputs are $S,\bold{\htau}$. We immediately have that  
\[
OPT\left(S,\bold{\htau}\right) \geq OPT\left(S,\left(\tau_i,\bold{\htau}_{-i}\right)\right)
\]
by virtue of the left-hand side optimization program being strictly more constrained. Since agent $i$'s privacy constraint is always satisfied when reporting his true threshold (by construction of Program~\eqref{eq:optim_threshold_final}), agent $i$ gets utility 
$g\left(OPT\left(S,\left(\tau_i,\bold{\htau}_{-i}\right)\right)\right)$. 
This is at least as high (by virtue of $g$ being decreasing) as the utility agent $i$ gets from misreporting as above, since then agent $i$ gets utility either $g\left(OPT\left(S,\bold{\htau}\right)\right)$ or $-\infty$ if his true privacy constraint is not satisfied.
\end{proof}

\begin{clm}
Fix any agent $i$. Reporting $\htau_i > \tau_i$ cannot increase agent $i$'s utility.
\end{clm}


\begin{proof}
Let us once again fix $S_{-i}$ and $\bold{\tau}_{-i}$. We have two cases for agent $i$:
\begin{itemize}
    \item In $\left(S,\bold{\htau}\right)$, agent $i$ receives $w_i \eta \leq \tau_i$. In this case, note that the optimization program with threshold $\tau_i$ and $\hat{\tau}_i$ are equivalent and
    \[
    OPT\left(S,\bold{\htau}\right) = OPT\left(S,\left(\tau_i,\bold{\htau}_{-i}\right)\right).
    \]
    In this case, agent $i$'s privacy constraint is satisfied and he gets utility 
    $g \left(OPT\left(S,\bold{\htau}\right)\right) = g\left(OPT\left(S,\left(\tau_i,\bold{\htau}_{-i}\right)\right)\right)$. I.e., his utility is unchanged. 
    \item Otherwise, agent $i$'s privacy constraint is not satisfied and he receives the worst possible utility of $-\infty$.
\end{itemize}
This concludes the proof. 
\end{proof}

\bibliographystyle{plainnat}
\bibliography{refs.bib}

\appendix

\section{Proofs of the Main Theorems}
\subsection{Proof of Theorem~\ref{thm:semi_close-form}}\label{app:proof_mainthm}

First, we consider the case in which Program~\ref{eq:weighted_optim_final} is not strictly feasible; i.e., there exists no feasible weights $\w$ such that for all $i$,
\[
c_i w_i \eta < f \left(\sigma^2 \sum_{i = 1}^n w_i^2 + \frac{2}{\eta^2}\right) - \outopt.
\]

\begin{clm} 
Suppose Program~\ref{eq:weighted_optim_final} is feasible, but has no strictly feasible solution. Then any feasible, hence the optimal solution is given by $w_i = K/c_i$ for all $i$ and for some constant $K$. 
\end{clm}

\begin{proof}
First, suppose the program is not strictly feasible. In particular, let us look at $\w$ the optimal solution to the program. We consider two cases:
\begin{enumerate}
\item There exists $i$ such that 
\[
c_i w_i \eta < f \left(\sigma^2 \sum_{i = 1}^n w_i^2 + \frac{2}{\eta^2}\right) - \outopt.
\]
Let $\tilde{w}_i = w_i + (n-1) \varepsilon$ and $\tilde{w}_j = w_j - \varepsilon$ for all $j \neq i$. Note that for $\varepsilon$ small enough, the $\tilde{\w}$ define proper weights: they are positive, since by Corollary~\ref{cor:full_participation}, $w_j > 0$ for all $j$, and they sum to $1$ by construction. Further, 
\[
c_j \tilde{w}_j \eta < c_j w_j \eta \leq f \left(\sigma^2 \sum_{i = 1}^n w_i^2 + \frac{2}{\eta^2}\right) - \outopt~~\forall j \neq i,
\]
hence the participation constraints are strict for all $j \neq i$. Finally, for $\varepsilon$ small enough, by continuity of $f$, we have 
\[
c_i \tilde{w_i} \eta < f \left(\sigma^2 \sum_{i = 1}^n \tilde{w}_i^2 + \frac{2}{\eta^2}\right) - \outopt,
\]
since the left-hand side converges to $c_i w_i \eta$ and the right-hand side to $f \left(\sigma^2 \sum_{i = 1}^n w_i^2 + \frac{2}{\eta^2}\right) - \outopt$ when $\varepsilon$ goes to $0$. Hence, $\tilde{\w}$ is a strictly feasible solution, which is a contradiction.
\item For all $i$, the participation constraint is tight, i.e. 
\[
c_i w_i \eta = f \left(\sigma^2 \sum_{i = 1}^n w_i^2 + \frac{2}{\eta^2}\right) - \outopt.
\]
Then $w_i = K/c_i$ for all $i$, where $K \triangleq \frac{f \left(\sigma^2 \sum_{i = 1}^n w_i^2 + \frac{2}{\eta^2}\right) - \outopt}{\eta}$ is the same for all agents $i$.
\end{enumerate}
\end{proof}

Now, we can consider without loss of generality the case in which the program is strictly feasible. In this case, Slater's condition holds and we can apply the KKT conditions. Note that the Lagrangian of Program~\eqref{eq:weighted_optim_fixedeta} is given by 
\begin{align*}
\mathcal{L}(\w,\vec{\lambda},\vec{\lambda^0},\gamma)
= \sigma^2 \sum_{i =1}^{n} w_i^2 + \frac{2}{\eta^2} + \gamma \left(1 - \sum_{i =1}^{n} w_i\right)
- \sum_{i=1}^{n} \lambda_i^0 w_i
+ \sum_{i = 1}^{n} \lambda_i \left( c_i w_i \eta - f\left(\sigma^2  \sum_{j = 1}^{n} w_j^2 + \frac{2}{\eta^2}\right) + \outopt)\right).
\end{align*}
The first order condition, taking the derivative with respect to $w_i$ (remembering that $f$ is differentiable by assumption), is given by 
\begin{align*}
0 
= 2 \sigma^2 w_i + \lambda_i c_i \eta - \gamma - \lambda_i^0
- 2 \sigma^2 \left(\sum_j \lambda_j\right) w_i \cdot f' \left(\sigma^2  \sum_{j = 1}^{n} w_j^2 + \frac{2}{\eta^2}\right).
\end{align*}
This implies $w_i = \frac{-\lambda_i c_i \eta + \gamma + \lambda_i^0}{2 \sigma^2 \left(1 - \left(\sum_j \lambda_j\right) \cdot  f' \left(\sigma^2  \sum_{j = 1}^{n} w_j^2 + \frac{2}{\eta^2}\right) \right)}$. 

Since $w_i > 0$ for all $i$, complementary slackness yields that $\lambda_i^0 = 0$ for all $i$. The first order condition then simplifies to 
\[
w_i = \frac{-\lambda_i c_i \eta + \gamma}{2 \sigma^2 \left(1 - \left(\sum_j \lambda_j\right) \cdot  f' \left(\sigma^2  \sum_{j = 1}^{n} w_j^2 + \frac{2}{\eta^2}\right) \right)}.
\]
We now have two cases:
\begin{enumerate}
\item Either $c_i w_i \eta = f \left(\sigma^2 \sum_{i = 1}^n w_i^2 + \frac{2}{\eta^2}\right) - \outopt$, i.e. the participation constraint is tight. Then we can write 
\[
w_i = \frac{f \left(\sigma^2 \sum_{i = 1}^n w_i^2 + \frac{2}{\eta^2}\right) - \outopt}{\eta c_i} \triangleq \frac{K}{c_i},
\]
where $K$ is a constant in that it is the same for all agents.
\item Either $w_i$ is such that $i$'s participation constraint is not tight. Then, by complementary slackness, we have that $\lambda_i = 0$, hence we can rewrite 
\begin{align*}&w_i 
 = \frac{\gamma}{2 \sigma^2 \left(1 - \left(\sum_j \lambda_j\right) \cdot  f' \left(\sigma^2  \sum_{j = 1}^{n} w_j^2 + \frac{2}{\eta^2}\right) \right)} 
 \triangleq W
\end{align*}
is a constant that is the same for all agents.
\end{enumerate}
Therefore, in any optimal solution, there exists constants $K$ and $W$ such that either $w_i = K/c_i$, or $w_i = W$. To conclude the proof, suppose that $i < j$, but $w_i$ satisfies case (1) above (and $w_i = K/c_i$) while $w_j$ satisfies case (2) (and $w_j = W$). We have that $w_j < K/c_j$ since the participation constraint $c_j w_j \leq K$ is not tight for $j$ by definition of case (2). Hence, $c_i w_i = K$ while $c_j w_j < K$, implying that $c_i w_i > c_j w_j$. This contradicts Claim~\ref{clm:monotone_equilibrium_costs}. Therefore, for any $j$, if $w_j = W$, it must be that $w_i = W$ for all $i < j$. This concludes the proof.

\subsection{Proof of Theorem~\ref{thm:semi_close-form_2}}\label{app:proof_mainthm_2}

First, we consider the case in which Slater's condition does not hold, and there exists at least one $i$ such that $w_i \eta = \tau_i$ in any feasible solution. We have two cases:
\begin{enumerate}
    \item There exists $j$ such that $w_j \eta < \tau_j$. Then let $\tilde{\w}$ be such that $\tilde{w}_i = w_i - \varepsilon$ for all $i \neq j$ and let $\tilde{w}_j = w_j + (n-1) \varepsilon$. When $\varepsilon$ is small enough, $\tilde{\w} \geq 0$ (noting that we have $\w > 0$ at an optimal solution by Claim~\ref{clm:positive_monotone}), the weights sum to $1$, $\tilde{w_i} \eta < (w_i \eta \leq) \tau_i$ for all $i \neq j$, and $\tilde{w}_j \eta < \tau_j$. Hence $\tilde{w}$ is strictly feasible, which is a contradiction.
    \item For all $i$, $w_i \eta = \tau_i$. Then the optimal solution is fully determined by these equations, and satisfies $w_i = \frac{\tau_i}{\eta}$ for all $i$. 
\end{enumerate}

Now, suppose we have strict feasibility, i.e. Slater's condition holds. The Lagrangian of the optimization problem is given by
\begin{align*}
\mathcal{L}(\w,\lambda,\lambda^0,\mu) 
= \sum_i w_i^2 \sigma^2 + \frac{2}{\eta^2} + \sum_i \lambda_i \left( w_i \eta - \tau_i \right) 
+ \mu \left(\sum_i w_i - 1\right) - \sum_i \lambda_i^0 w_i.
\end{align*}
The first order condition (with respect to agent $i$) is then given by 
\[
2 w_i \sigma^2 + \lambda_i \eta + \mu -\lambda_i^0 = 0,
\]
which implies
\[
w_i = \frac{\lambda_i \eta + \mu - \lambda_i^0}{2 \sigma^2}.
\]
By Claim~\ref{clm:positive_monotone}, $w_i > 0$, hence by KKT conditions, $\lambda_i^0 = 0$. Therefore, we can rewrite 
\[
w_i = \frac{\lambda_i \eta + \mu}{2 \sigma^2}.
\]
We now have two cases: 
\begin{enumerate}
\item Either agent $i$'s privacy constraint is tight. Then, $w_i \eta = \tau_i$, i.e. $w_i = \frac{\tau_i}{\eta}$.
\item Otherwise, the privacy constraint is not tight. Then, by the KKT conditions, it must be that $\lambda_i = 0$, hence $w_i = \frac{\mu}{2 \sigma^2} = W$ for some constant $W$. Since the privacy constraint is not tight, we have in particular that $W < \frac{\tau_i}{\eta}$.
\end{enumerate}

To complete the proof, suppose that we have $i < j$ such that $i$ is in case (1) and $w_i = \tau_i/\eta$, while $j$ is in case (2) and $w_j = W$. We have that $w_i/\tau_i \leq w_j/\tau_j$ by Claim~\ref{clm:positive_monotone}. This then implies that $1/\eta \leq W/\tau_j$, and in turn that $\tau_j \leq W \eta < \tau_j$ (remember that since $j$ is in case (2), $W \eta < \tau_j$). This is a contradiction. Hence, it must be the case that if if $w_i$ is in case (1), we must have $w_j = w_j/\eta$ for all subsequent $j > i$.

\section{Proofs of Supporting Claims}
\subsection{In the Quasi-Linear Utility Model}\label{app:proof_claims_quasilinear}

\subsubsection{Proof of Claim~\ref{clm: optim_final}}
Pick any $\w, \eta$, and let $S = \{i~\text{s.t.}~w_i > 0\}$. We first note that
\[
\sigma^2 \sum_{i \in S} w_i^2 + \frac{2}{\eta^2} = \sigma^2 \sum_{i \in [n]} w_i^2 + \frac{2}{\eta^2}
\]
since $w_i = 0$ for all $i \notin S$. Hence, $(\w,\eta,S)$ achieves the same objective value for Program~\ref{eq:weighted_optim} as $(\w,\eta)$ for Program~\ref{eq:weighted_optim_final} for any $\w, \eta$, and $S$ constructed as above. 

Second, with respect to Program~\ref{eq:weighted_optim}, we have that:
 \begin{enumerate}
 \item $\sum_{i \in S} w_i = 1 \Leftrightarrow \sum_{i \in [n]} w_i = 1$ by virtue of $w_i = 0$ for all $i \notin S$.
 \item Since we have that 
 \[
 f \left(\sigma^2 \sum_{i \in [n]} w_i^2 + \frac{2}{\eta^2}\right) = f \left(\sigma^2 \sum_{i \in S} w_i^2 + \frac{2}{\eta^2}\right),
 \]
 for all $i$, $c_i w_i \eta \leq f \left(\sigma^2 \sum_{i \in [n]} w_i^2 + \frac{2}{\eta^2}\right) - \outopt$ if and only if $c_i w_i \eta \leq f \left(\sigma^2 \sum_{i \in S} w_i^2 + \frac{2}{\eta^2}\right) - \outopt$. Therefore $(\w,\eta,S)$ is feasible for Program~\ref{eq:weighted_optim} if and only $(\w,\eta)$ is feasible for Program~\ref{eq:weighted_optim_final}. 
\end{enumerate}

 This is enough to conclude the proof. Indeed, since $(\w,\eta)$ feasible for Program~\eqref{eq:weighted_optim_final}$(\w,\eta,S)$ implies $(\w,\eta,S)$ feasible for Program~\ref{eq:weighted_optim} and they both have the same objective value, the optimal value of Program~\eqref{eq:weighted_optim} is at least that of Program~\eqref{eq:weighted_optim_final}. Vice-versa, the optimal value of Program~\ref{eq:weighted_optim_final} is at least that of Program~\ref{eq:weighted_optim}. Hence, Program~\ref{eq:weighted_optim} and Program~\ref{eq:weighted_optim_final} have the same optimal value. Further, if $(\w,\eta)$ is optimal, then $(\w,\eta,S)$ is optimal by virtue of having the same objective value, and vice-versa. This concludes the proof. 
 
\subsubsection{Proof of Claim~\ref{clm:monotone_w}}
 
Let $\w$ be an optimal solution to Program~\eqref{eq:weighted_optim_fixedeta}. Suppose there exists $i < j$ such that $w_i < w_j$. Now, let us look at a possible alternative solution $\tilde{\w}$ where $\tw_i \triangleq w_i + \varepsilon,~\tw_j \triangleq w_j - \varepsilon$ for $\varepsilon > 0$ small enough, and $\tw_{k} \triangleq w_{k}$ for all agents $k \neq i,j$. We will show that this solution leads to a smaller objective, contradicting optimality. 

First, 
\begin{align*}
\sum_{k=1}^n \tw_k^2 - \sum_{k=1}^n w_k^2
&= (w_i + \varepsilon)^2 - w_i^2 + (w_j - \varepsilon)^2 - w_j^2
\\&= 2 w_i \varepsilon + \varepsilon^2 + (\varepsilon^2 - 2w_j \varepsilon)
\\&= 2 \varepsilon (w_i - w_j + \varepsilon)
\\& < 0,
\end{align*}
for $\varepsilon$ small enough, as $w_i < w_j$. This shows that $\tw$ leads to a better variance than $\w$, since it directly implies 
\[
\sum_{k = 1}^n \tw_k^2 \sigma^2 + \frac{2}{\eta^2} < \sigma^2 \sum_{k = 1}^n w_i^2+ \frac{2}{\eta^2}.
\]

Second, since $c_i \leq c_j$, $\tw_j < w_j$, for $\varepsilon$ small enough we have $\tw_i \leq \tw_j$, and because $f$ is decreasing, it follows that
\begin{align*}
c_i \tw_i \eta \leq c_i \tw_j \eta <  c_j w_j \eta 
\leq f \left(\sigma^2 \sum_{i = 1}^n w_i^2 + \frac{2}{\eta^2}\right)  
 \leq f \left(\sigma^2 \sum_{i = 1}^n \tw_i^2 + \frac{2}{\eta^2}\right).
\end{align*}
Further, because $\tw_j < w_j$, we have
\begin{align*}
c_j \tw_j \eta  < c_j w_j \eta 
\leq f \left(\sigma^2 \sum_{i = 1}^n w_i^2 + \frac{2}{\eta^2}\right)  
\leq f \left(\sigma^2 \sum_{i = 1}^n \tw_i^2 + \frac{2}{\eta^2}\right).
\end{align*}
Therefore, $\tilde{\w}$ is feasible, since it satisfies the participation constraints and that the weights are still positive and sum to $1$. This concludes the proof.
 
\subsubsection{Proof of Claim~\ref{clm:monotone_equilibrium_costs}}
 
Let $\w$ be an optimal solution with $c_i w_i \eta  > c_j w_j \eta$; in particular, it must be that $w_i > w_j$ since $c_i \leq c_j$. For small enough $\varepsilon$, let $\tw_i = w_i - \varepsilon$, $\tw_j = w_j + \varepsilon$, and $\tw_k = w_k$ i for any other agent $k \neq i,j$. First, we note that this transformation decreases the variance. Indeed,  
\begin{align*}
w_i^2 + w_j^2 - \tw_i^2 - \tw_j^2 
&=  w_i^2 + w_j^2 -w_i^2 + 2 \varepsilon w_i - \varepsilon^2 - w_j^2 - 2\varepsilon w_j -\varepsilon^2 
\\&= 2 \varepsilon(w_i - w_j) - 2\varepsilon^2
\\&= 2 \varepsilon (w_i - w_j - \varepsilon)
\\&> 0
\end{align*}
when $\varepsilon$ is small enough, by virtue of $w_i > w_j$. Further, the constraints that the weights must sum to $1$ still holds, as well as the non-negativity constraint so long as $\varepsilon$ is small enough (smaller than $w_i$). Finally, $c_i \tw_i \eta = c_i (w_i - \varepsilon) \eta \leq c_i w_i \eta$, and $c_j \tw_j \eta = c_j (w_j \eta + \varepsilon \eta) \leq c_i w_i \eta$ for small enough $\varepsilon$ (as $c_j w_j \eta  < c_i w_i \eta$); combining this with the fact that the variance decreases, the participation constraints still holds. Therefore, $\tilde{\w}$ is feasible for Program~\ref{eq:weighted_optim_final} and has strictly better objective value than an optimal solution, which is a contradiction.

\subsection{In the Privacy-Constrained Utility Model}
\subsubsection{Proof of Claim~\ref{clm:positive_monotone}}

We show the results in the claim by contradiction. First, suppose there exists $i$ such that $w_i = 0$. We will show that we can construct alternative weight vector $\w$ that is feasible and leads to a strictly better objective value, contradicting optimality of $\w$. To do so, let $j$ be such that $w_j > 0$, and let $\tw_i = \varepsilon$, $\tw_j = w_j - \varepsilon$, and $\tw_j = w_j$ for all $k \neq i,j $. For $\varepsilon$ small enough, note that $\tilde{\w}$ is feasible: all weights remain non-negative, sum to $1$, $\tw_j \eta \leq w_j \eta \leq \tau_j$, and $\tw_i \eta = \varepsilon \eta \leq \tau_i$ so as long as $\varepsilon$ is sufficiently small. Further, the objective value under $\tilde{\w}$ is smaller than under $\w$. Indeed, the change in variance (renormalized by $1/\sigma^2$) is given by
\begin{align*}
w_i^2 - \tw_i^2 + w_j^2 - \tw_j^2
&= -\varepsilon^2 + w_j^2 - (w_j - \varepsilon)^2
\\&= -\varepsilon^2 + w_j^2  - w_j^2 + 2 w_j \varepsilon - \varepsilon^2 
\\&=2 \varepsilon (w_j - \varepsilon)
\\&> 0.
\end{align*}
where the last inequality follows from $\varepsilon$ being small enough. This is a contradiction.

Second, suppose there exists $i < j$ such that $w_i < w_j$. Let us construct as before an alternative weight vector $\tilde{\w}$ such that $\tw_i = w_i + \varepsilon$, $\tw_j = w_j - \varepsilon$, and $\tw_k = w_k$ for all $k \neq i,j$. First, $\tilde{\w}$ is feasible: $\sum_i \tw_i = 1$, $\tw_j \geq 0$ for $\varepsilon$ small enough (since $w_j > 0$), $\tw_i \eta <  w_j \eta \leq \tau_j \leq \tau_i$ for $\varepsilon$ small enough (as $w_i < w_j$), and $\tw_j \eta \leq w_j \eta \leq \tau_j$. Further, $\tilde{\w}$ yields better variance than $\w$. Indeed, 
\begin{align*}
\sum_{i =1}^n w_i^2 - \sum_{i =1}^n \tw_i^2 
&= w_i^2 - \tw_i^2 + w_j^2 - \tw_j^2
\\&= w_i^2 - (w_i + \varepsilon)^2 + w_j^2 - (w_j - \varepsilon)^2 
\\&= w_i^2  - w_i^2  - \varepsilon^2 - 2 w_i \varepsilon + w_j^2  - w_j^2 + 2 w_j \varepsilon - \varepsilon^2 
\\&=2 \left(w_j - w_i - \varepsilon\right)
\\&> 0
\end{align*}
for $\varepsilon$ small enough, remembering that $w_j > w_i$. This is a contradiction.

Finally, suppose there exists $i < j$ such that $\frac{w_i}{\tau_i} > \frac{w_j}{\tau_j}$ (note that since $\tau_i \geq \tau_j$, this also implies $w_i > w_j$). Then, consider alternative weight vector $\tw_i = w_i - \varepsilon$, $\tw_j = w_j + \varepsilon$, and $\tw_k = w_k$ for all $k \neq i,j$. First, we note that $\w$ is feasible. Indeed, $\sum_i \tw_i = 1$, $\tw_i \eta < w_i \eta \leq \tau_i$, and for $\varepsilon$ small enough,
\[
\frac{\tw_j}{\tau_j} \eta 
< \frac{\tw_i}{\tau_i} \eta
< \frac{w_i}{\tau_i}  \eta
\leq 1.
\]
$\tilde{\w}$ also has lower variance than $\w$, by a similar calculation as before:
\begin{align*}
\sum_{i =1}^n w_i^2 - \sum_{i =1}^n \tw_i^2
&=  w_i^2 - \tw_i^2 + w_j^2 - \tw_j^2
\\&= w_i^2 - (w_i - \varepsilon)^2 + w_j^2 - (w_j + \varepsilon)^2 
\\&= w_i^2  - w_i^2  - \varepsilon^2 + 2 w_i \varepsilon + w_j^2  - w_j^2 - 2 w_j \varepsilon - \varepsilon^2  
\\&=2 \varepsilon \left(w_i - w_j - \varepsilon\right)
\\&> 0
\end{align*}
where the last step follows from $w_i > w_j$ and $\varepsilon$ small enough. This is a contradiction.

\subsubsection{Proof of Claim~\ref{clm:simple_optim_eta}}\label{app:simple_optim_eta}
We have that 
\begin{align*}
f'(\eta) 
&= \frac{2 \sigma^2  \sum_{i = t+1}^n \tau_i }{t \eta^2} \left(1 - \frac{1}{\eta}  \sum_{i = t+1}^n \tau_i \right) - \frac{2\sigma^2}{\eta^3}  \sum_{i = t+1}^n \tau_i^2 - \frac{4}{\eta^3}
\\&= \frac{2 \sigma^2  \sum_{i = t+1}^n \tau_i }{t \eta^3}\left(\eta -  \sum_{i = t+1}^n \tau_i  \right) - \frac{2\sigma^2}{\eta^3}  \sum_{i = t+1}^n \tau_i^2 - \frac{4}{\eta^3}
\\&=\frac{2 \sigma^2}{\eta^3} \left(\frac{\eta}{t} \sum_{i = t+1}^n \tau_i  - \frac{1}{t} \left(\sum_{i = t+1}^n \tau_i\right)^2 - \sum_{i = t+1}^n \tau_i^2 - \frac{2}{\sigma^2}\right).
\end{align*}
In turn, $f'(\eta) < 0$ if and only if $\eta < \frac{\left(\sum_{i = t+1}^n \tau_i\right)^2 + t \sum_{i = t+1}^n \tau_i^2 + 2t/\sigma^2}{\sum_{i = t+1}^n \tau_i}$ and $f'(\eta) > 0$ if and only if $\eta < \frac{\left(\sum_{i = t+1}^n \tau_i\right)^2 + t \sum_{i = t+1}^n \tau_i^2 + 2/\sigma^2}{\sum_{i = t+1}^n \tau_i}$. Finally, note that $f'(\eta) = 0$ can be written as 
\[
\frac{\eta}{t} \sum_{i = t+1}^n \tau_i  - \frac{1}{t} \left(\sum_{i = t+1}^n \tau_i\right)^2 - \sum_{i = t+1}^n \tau_i^2 - \frac{2}{\sigma^2} = 0.
\]
This immediately leads to 
\[
\eta^* = \frac{\left(\sum_{i = t+1}^n \tau_i\right)^2 + t \sum_{i = t+1}^n \tau_i^2 + \frac{2t}{\sigma^2}}{\sum_{i = t+1}^n \tau_i}.
\]

\end{document}